\documentclass[11pt]{article}
\usepackage{amsmath,amsthm,tikz,graphicx,fix-cm}
\usepackage[longend,noresetcount,boxed]{algorithm2e}
\usepackage{latexsym,epsfig,amssymb,url}
\usepackage{verbatim}
\usetikzlibrary{arrows}

\newtheorem{lemma}{Lemma}

\newtheorem{theorem}{Theorem}

\newenvironment{proof_of_claim}{\noindent {\it Proof of Claim: }}{\hspace*{\fill} $\diamond$}

\title{Improved Approximations for Cubic Bipartite and Cubic TSP}

\author{ Anke van Zuylen\thanks{Supported in part by NSF Prime Award: HRD-1107147, Women in Scientific Education (WISE) and by a grant from the Simons Foundation (\#359525, Anke Van Zuylen).}\\ 
\small Department of Mathematics\\ 
\small  The College of William and Mary, Williamsburg, VA, 23185, USA\\
\small     \texttt{anke@wm.edu}
}
\begin{document}
\date{\vspace{-2\baselineskip}}
\maketitle

\DontPrintSemicolon
\LinesNotNumbered
\SetAlCapSkip{8pt}

\begin{abstract}
We show improved approximation guarantees for the traveling salesman problem on cubic bipartite graphs and cubic graphs.
For cubic bipartite graphs with $n$ nodes, we improve on recent results of Karp and Ravi by giving a ``local improvement'' algorithm that finds a tour of length at most $5/4n-2$. For 2-connected cubic graphs, we show that the techniques of  M\"omke and Svensson can be combined with the techniques
of Correa, Larr\'e and Soto, to obtain a tour of length at most $(4/3-1/8754)n$.

\end{abstract}

\section{Introduction}
The traveling salesman problem (TSP) is one of the most famous and widely studied combinatorial optimization problems. Given a set of cities and pairwise distances, the goal is to find a tour of minimum length that visits every city exactly once. Even if we require the distances to form a metric, the problem remains NP-hard. The classic Christofides' algorithm~\cite{Christofides76} finds a tour that has length at most $\frac32$ times the length of the optimal tour. Despite much effort in the 35 years following Christofides's result, we do not know any algorithms that improve on this guarantee. 

One approach that has often been useful in designing approximation algorithms is the use of linear programming. In this context, a major open question is to determine the integrality gap of the subtour elimination linear program or Held-Karp relaxation~\cite{DantzigFJ54,HeldK70}; the integrality gap is the worst-case ratio of the length of an optimal tour to the optimal value of the relaxation. Examples are known in which the length of the optimal tour is $\frac43$ times the value of the Held-Karp relaxation, and a major open question is whether this is tight.

Recent years have seen some exciting progress towards answering this question on graph metrics, also called the graph-TSP.
In this special case of the metric TSP, we are given an unweighted graph $G=(V,E)$ in which the nodes represent the cities, and the distance between two cities is equal to the shortest path in $G$ between the corresponding nodes. Examples are known in which the ratio between the length of the optimal tour and the Held-Karp relaxation is $\frac 43$, where the examples are in fact graph-TSP instances with an underlying graph $G$ that is 2-connected and subcubic (every node has degree at most three).

The graph-TSP thus captures many of the obstacles that have prevented us from obtaining improved approximations for general metrics, and much recent research has focused on finding improved algorithms for the graph-TSP.
The first improvement for graph-TSP metrics is due to Gamarnik, Lewenstein and Sviridenko~\cite{GamarnikLS05}, who show an approximation guarantee strictly less than $\frac32$ for cubic, 3-connected graphs. Aggarwal, Garg and Gupta~\cite{AggarwalGG11} give a $\frac43$-approximation algorithm for this case.
Boyd, Sitters, Van der Ster, and Stougie~\cite{BoydSSS14} show that there is a $\frac43$-approximation algorithm for any cubic graph, and M\"omke and Svensson~\cite{MomkeS16} show this holds also for subcubic graphs.

M\"omke and Svensson also show a 1.461-approximation algorithm if we make no assumptions on the underlying graph $G$. Mucha~\cite{Mucha14} improves their analysis to show an approximation guarantee of $\frac{13}9$. Seb\H o and Vygen~\cite{SeboV14} combine the techniques of M\"omke and Svensson with a clever use of ear decompositions, which yields an approximation ratio of $1.4$.

As mentioned previously, for subcubic graphs, examples exist that show that we cannot obtain better approximation guarantees than $\frac43$ unless we use a stronger lower bound on the optimum than the Held-Karp relaxation or subtour elimination linear program. Correa, Larr\'e and Soto~\cite{CorreaLS15} show that this is not the case for cubic graphs. They refine the techniques of Boyd et al.~\cite{BoydSSS14} and show how to find a tour of length at most $\left(\frac43-\frac1{61236}\right)n$ for the graph-TSP on a 2-connected cubic graph $G$, where $n$ is the number of nodes.
Correa, Larr\'e and Soto also consider the graph-TSP on planar cubic bipartite 3-connected graphs, and give a $(\frac43-\frac1{18})$-approximation algorithm. Planar cubic bipartite 3-connected graphs are known as Barnette graphs, and a long-standing conjecture in graph theory by Barnette~\cite{Barnette69} states that all planar cubic bipartite 3-connected graphs are Hamiltonian.
Recently, Karp and Ravi~\cite{KarpR14} gave a $\frac97$-approximation algorithm for the graph-TSP on a superset of Barnette graphs, cubic bipartite graphs. 

In this paper, we give two results that improve on the results for cubic graph-TSP. For the graph-TSP on (non-bipartite) cubic graphs, we show that the techniques of M\"omke and Svensson~\cite{MomkeS16} can be combined with those of Correa et al.~\cite{CorreaLS15} to find an approximation algorithm with guarantee $(\frac43-\frac1{8754})$. We note that independent of our work, Candr\'akov\'a and Lukot'ka~\cite{CandrakovaL15} showed very recently, using different techniques, how to obtain a $1.3$-approximation algorithm for the graph-TSP on cubic graphs.
For connected bipartite cubic graphs, we give an algorithm that finds a tour of length at most $\frac{5}{4} n-2$, where $n$ is the number of nodes. The idea behind our algorithm is the same as that of many previous papers, namely to find a cycle cover (or 2-factor) of the graph with a small number of cycles. Our algorithm is basically a simple ``local improvement'' algorithm. The key idea for the analysis is to assign the size of each cycle to the nodes contained in it in a clever way; this allows us to give a very simple proof that the algorithm returns a 2-factor with at most $n/8$ components.
We also give an example that shows that the analysis is tight, even if we relax a certain condition in the algorithm that restricts the cases when we allow the algorithm to move to a new solution.

The remainder of this paper is organized as follows.
In Section~\ref{sec:bip}, we describe and analyze our algorithm for the graph-TSP on cubic bipartite graphs, and in Section~\ref{sec:cub}, we give our improved result for non-bipartite cubic graphs. 

\section{The Graph-TSP on Cubic Bipartite Graphs}\label{sec:bip}

In the graph-TSP, we are given a graph $G=(V,E)$, and for any $u,v\in V$, we let the distance between $u$ and $v$ be the number of edges in the shortest path between $u$ and $v$ in $G$. The goal is to find a tour of the nodes in $V$ that has minimum total length.
A 2-factor of $G$ is a subset of edges $F\subseteq E$, such that each node in $V$ is incident to exactly two edges in $F$. Note that if $F$ is a 2-factor, then each (connected) component of $(V,F)$ is a simple cycle. If $C$ is a component of $(V, F)$, then we will use $V(C)$ to denote the nodes in $C$ and $E(C)$ to denote the edges in $C$. The size of a cycle $C$ is defined to be $|E(C)|$ (which is of course equal to $|V(C)|$). Sometimes, we consider a component of $(V, F\setminus E')$ for some $E'\subset F$. A component of such a graph is either a cycle $C$ or a path $P$. We define the length of a path $P$ to be the number of edges in $P$. 

The main idea behind our algorithm for the graph-TSP in cubic bipartite graphs (and behind many algorithms for variants of the graph-TSP given in the literature) is to find a 2-factor $F$ in $G$ such that $(V,F)$ has a a small number of cycles, say $k$. We can then contract each cycle of the 2-factor, find a spanning tree on the contracted graph, and add two copies of the corresponding edges to the 2-factor. This yields a spanning Eulerian (multi)graph containing $n+2(k-1)$ edges. By finding a Eulerian walk in this graph and shortcutting, we get a tour of length at most $n+2k-2$. In order to get a good algorithm for the graph-TSP, we thus need to show how to find a 2-factor with few cycles, or, equivalently, for which the average size of the cycles is large.

In Section~\ref{tech}, we give an algorithm for which we prove in Lemma~\ref{lemma:2factor} that, given a cubic bipartite graph $G=(V, E)$, it returns a 2-factor with average cycle size at least 8. By the arguments given above, this implies the following result.
\begin{theorem}
There exists a $\frac54$-approximation algorithm for the graph-TSP on cubic bipartite graphs.
\end{theorem}

Before we give the ideas behind our algorithm and its analysis in Section~\ref{local}, we begin with the observation that we may assume without loss of generality that the graph has no ``potential 4-cycles'': a set of 4 nodes $S$ will be called a potential 4-cycle if there exists a 2-factor in $G$ that contains a cycle with node set exactly $S$. The fact that we can modify the graph so that $G$ has no potential 4-cycles was also used by Karp and Ravi~\cite{KarpR14}. 
\begin{lemma}\label{lemma:four}
To show that every simple cubic bipartite graph $G=(V, E)$ has a 2-factor with at most $|V|/8$ components, it suffices to show that every simple cubic bipartite graph $G'=(V',E')$ with no potential 4-cycles has a 2-factor with at most $|V'|/8$ components.
\end{lemma}
\begin{proof}
We show how to contract a potential 4-cycle $S$ in $G$ to get a simple cubic bipartite graph $G'$ with fewer nodes than the original graph, and how, given a 2-factor with average component size 8 in $G'$, we can uncontract $S$ to get a 2-factor in $G$ without increasing the number of components.

Let $S=\{v_1,v_2,v_3,v_4\}$ be a potential 4-cycle in $G$, i.e., $E[S]$ contains 4 edges, say $\{v_1,v_2\}, \{v_2, v_3\}, \{v_3, v_4\}, \{v_1, v_4\}$, and there exists no node $v_5\not\in S$ that is incident to two nodes in $\{v_1,v_2,v_3,v_4\}$ (since in that case a 2-factor containing a cycle with node set $S$ would have $v_5$ as an isolated node, and this cannot be a 2-factor since $v_5$ must have degree 2 in a 2-factor).

\newcommand{\odd}{\mathrm{odd}}
\newcommand{\even}{\mathrm{even}}
We contract $S$, by identifying $v_1,v_3$ to a new node $v_{\odd}$, and identifying $v_2,v_4$ to a new node $v_{\even}$. We keep a single copy of the edge $\{v_{\odd}, v_{\even}\}$. The new graph $G'$ is simple, cubic and bipartite, and $|V'|= |V|-2$. See Figure~\ref{fig:square} for an illustration.

\begin{figure}[ht!]\begin{center}
\resizebox{1.5in}{1.5in}{
\begin{tikzpicture}[->,>=stealth',shorten >=1pt,auto,node distance=3cm,
  thick,main node/.style={circle,fill=blue!20,draw,font=\sffamily\LARGE\bfseries},aux node/.style={}]

  \node[main node] (1) {$v_1$};
  \node[main node] (2) [right of=1] {$v_2$};
  \node[main node] (3) [below of=2] {$v_3$};
  \node[main node] (4) [below of=1] {$v_4$};
  \node[aux node] (1a) [above left of=1] {};
  \node[aux node] (2a) [above right of=2] {};
  \node[aux node] (3a) [below right of=3] {};
  \node[aux node] (4a) [below left of=4] {};

   \path[-,every node/.style={font=\sffamily\LARGE}]
     (1a) edge node [left] {$e_1$} (1)
     (2a) edge node [left] {$e_2$} (2)
     (4a) edge node [left] {$e_4$} (4)
     (3a) edge node [left] {$e_3$} (3)
     (1) edge node [left] {} (2)
     (2) edge node [left] {} (3)
     (3) edge node [left] {} (4)
     (4) edge node [left] {} (1)
;
\end{tikzpicture}
}\hspace*{1in}
\resizebox{1.5in}{1.5in}{

\begin{tikzpicture}[->,>=stealth',shorten >=1pt,auto,node distance=3cm,
  thick,main node/.style={circle,fill=blue!20,draw,font=\sffamily\LARGE\bfseries},aux node/.style={}]

  \node[main node] (5) {$v_{\mathrm{odd}}$};
  \node[main node] (6) [below of=5] {$v_{\mathrm{even}}$};
  \node[aux node] (7) [right of=5] {};
  \node[aux node] (8) [right of=6] {};
  \node[aux node] (1a) [above left of=5] {};
  \node[aux node] (2a) [above right of=7] {};
  \node[aux node] (3a) [below right of=8] {};
  \node[aux node] (4a) [below left of=6] {};

   \path[-,every node/.style={font=\sffamily\LARGE}]
     (1a) edge node [left] {$e_1$} (5)
     (2a) edge node [left] {$e_2$} (6)
     (4a) edge node [left] {$e_4$} (6)
     (3a) edge node [left] {$e_3$} (5)
     (5) edge node [left] {} (6)
;
\end{tikzpicture}
}
\end{center}
\caption{The 4-cycle on the left is contracted, by identifying $v_1,v_3$ to a new node $v_{\odd}$, and identifying $v_2,v_4$ to a new node $v_{\even}$, and keeping a single copy of the edge $\{v_{\odd}, v_{\even}\}$, to obtain the simple cubic bipartite graph on the right.}\label{fig:square}
\end{figure}
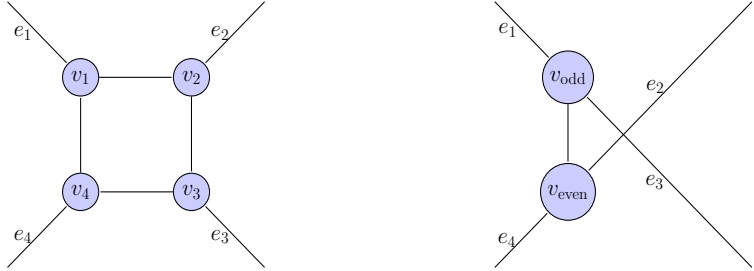

Given any 2-factor in $G'$, we can ``uncontract'' $S$ and find a 2-factor in $G$ with at most as many components as the 2-factor in $G'$:
If  the 2-factor on $G'$ does not contain $\{v_{\odd},v_{\even}\}$ then it must contain the other 4 edges incident to $v_\odd$ and $v_\even$. When uncontracting $S$, this gives one edge incident to each $v_i, i=1, \ldots, 4$. Since all other node degrees are even, the graph consists of even cycles and two paths with endpoints in $\{v_1,v_2,v_3,v_4\}$. We can choose to add the edges $\{v_1,v_2\}, \{v_3,v_4\}$, or the edges $\{v_2,v_3\}, \{v_1,v_4\}$; both of these choices give a 2-factor, and at least one of the two options must give a 2-factor in which all 4 nodes in $S$ are in the same cycle.
If  the 2-factor on $G'$ does contain $\{v_{\odd},v_{\even}\}$ then it must contain one other edge incident to $v_\odd$ and to $v_\even$. When uncontracting $S$, this gives one edge incident to $v_1$ or $v_3$, and one edge incident to $v_2$ or $v_4$. Suppose without loss of generality the edges are incident to $v_1$ and $v_2$. Then, we add edges $\{v_2,v_3\}, \{v_3,v_4\}$, and $\{v_1,v_4\}$ to get a 2-factor. Note that it is again the case that all 4 nodes in $S$ are in the same cycle in the resulting 2-factor.
\end{proof}

\subsection{A Local Improvement Heuristic}\label{local}
A cubic bipartite graph has a perfect matching (in fact, it is the case that the edge set can be decomposed into three perfect matchings), and given a cubic bipartite graph $G=(V,E)$, we can obtain a 2-factor $F$ by simply finding a perfect matching $M$ and letting $F=E\backslash M$. Conversely, if $F$ is a 2-factor for $G$, then $E\backslash F$ is a perfect matching in $G$. 
Now, given an arbitrary 2-factor $F_1$, we can use these observations to build a second 2-factor $F_2$ such that most nodes that are in a small cycle in $(V,F_1)$ are in a long cycle in $(V,F_2)$: The 2-factor $F_2$ is constructed by taking the matching $E\backslash F_1$ and adding half of the edges from each cycle in $(V,F_1)$. Since each cycle is even, its edges can be decomposed into two perfect matchings, and we may choose either one of them to add to $F_2$. We will say 2-factor $F_2$ is {\it locally optimal with respect to $F_1$} if $F_2$ contains all edges in $E\setminus F_1$ and for each cycle $C$ in $(V,F_1)$, replacing $F_2$ by $F_2\triangle E(C)$ does not reduce the number of components of $(V,F_2)$ (where $\triangle$ denotes the symmetric difference operator).

The essence of our algorithm is to start with an arbitrary 2-factor $F_1$, find a 2-factor $F_2$ that is locally optimal with respect to $F_1$, and return the 2-factor among $F_1,F_2$ with the smallest number of components.

If we consider a 6-cycle $C$ in $(V,F_1)$, and a 2-factor $F_2$ that is locally optimal with respect to $F_1$, then it is not hard to see that at least two edges of $C$ will be part of the same cycle, say $D$, in $(V,F_2)$. Moreover, the fact that the graph $G$ has no potential 4-cycles can be shown to imply that $D$ has size at least 10. 
This observation motivates the condition in Lemma~\ref{lemma:avg} below that for any $C$ in $(V, F_1)$, there should exist $D$ in $(V, F_2)$ of size at least 10, such that $|V(C)\cap V(D)|\ge 4$. 

In Lemma~\ref{lemma:avg} we show that this condition suffices to guarantee that either $(V,F_1)$ or $(V,F_2)$ has at most $|V|/8$ cycles. In Lemma~\ref{lemma:chordless} we show the condition holds for $F_2$ that is locally optimal for $F_1$, provided that all cycles in $F_1$ are {\em chordless}:
An edge $\{x,y\}$ is a {\it chord} for cycle $C$ in $(V, F_1)$, if $x,y\in C$, and $\{x,y\}\in E\setminus F_1$; a cycle will be referred to as chorded if it has at least one chord, and chordless otherwise.

A few more details are needed to deal with the general case when $F_1$ is not necessarily chordless; these are postponed to Section~\ref{tech}.

\begin{lemma}\label{lemma:avg}
Let $G=(V,E)$ be a simple cubic bipartite graph that has no potential 4-cycles, let $F_1$ and $F_2$ be 2-factors in $G$, such that
for any cycle $C$ in $(V, F_1)$, there exists a cycle $D$ in $(V, F_2)$ of size at least 10 such that $|V(C)\cap V(D)|\ge 4$.
Then either $(V, F_1)$ or $(V, F_2)$ has at most $|V|/8$ components.
\end{lemma}
\begin{proof}
Let $K_i$ be the number of components of $(V, F_i)$ for $i=1,2$. Note that it suffices to show that $\gamma K_1+(1-\gamma)K_2\le |V|/8$ for some $0\le \gamma\le1$.

In order to do this, we introduce a value $\alpha(v)$ for each node.
This value is set based on the size of the cycle containing $v$ in the second 2-factor, and they will satisfy $\sum_{v\in D} \alpha(v)= 1$ for every cycle $D$ in the second 2-factor $(V, F_2)$. Hence, we have that $\sum_{v\in V}\alpha(v)$ is equal to the number of cycles in $(V, F_2)$.
We will then show that the condition of the lemma guarantees that for a cycle $C$ in $(V, F_1)$,
\begin{equation}\sum_{v\in C}\alpha (v) \le \frac 16 |V(C)|-\frac13.\tag{*}\label{star}\end{equation}
This suffices to prove what we want: we have $K_2=\sum_{v\in V} \alpha(v)\le \frac 16|V| - \frac 13 K_1$, 
which is the same as $\frac14K_1+\frac34K_2\le \frac 18|V|$. 

The basic idea to setting the $\alpha$-values is that if 
$v$ is in a cycle $D$ in $(V, F_2)$ of size $k$, then we have $\alpha(v)=\frac1 k$.  The only exception to this rule is when $D$ has size 10; in this case we set $\alpha(v)$ for $v\in D$ to either $\frac16$ or $\frac1{12}$. There will be exactly $8$ nodes with $\alpha(v)=\frac1{12}$ and 2 nodes with $\alpha(v)=\frac16$. The nodes $v$ in $D$ with $\alpha(v)=\frac1{12}$ are chosen in such a way that, if there is a cycle $C$ in $(V, F_1)$ containing at least 4 nodes in $D$, then at least $4$ of the nodes in $V(C)\cap V(D)$ will have $\alpha(v)=\frac 1{12}$. 
It is possible to achieve this, since the fact that $D$ has 10 nodes implies that $(V, F_1)$ can contain at most two cycles that intersect $D$ in 4 or more nodes.

It is easy to see that (\ref{star}) holds: by the condition in the lemma, any cycle $C$ contains at least 4 nodes $v$ such that $\alpha(v)\le \frac 1 {12}$. Since we assumed in addition that $G$ has no potential 4-cycles, we also know that $\alpha(v)\le \frac16$ for all other $v\in V(C)$. Hence
$\sum_{v\in C}\alpha(v)\le \frac 16 (|V(C)|-4)+4\cdot\frac1{12} = \frac16|V(C)|-\frac13$.
\end{proof}

By Lemma~\ref{lemma:avg}, it is enough to find a 2-factor $F_2$ that satisfies that every cycle in the first 2-factor, $F_1$, has at least 4 nodes in some ``long'' cycle of $F_2$ (where ``long'' is taken to be size 10 or more).
The following lemma states that a locally optimal $F_2$ satisfies this condition, provided that $F_1$ is chordless. 

\begin{lemma}\label{lemma:chordless}
Let $G=(V,E)$ be a simple cubic bipartite graph that has no potential 4-cycles, let $F_1$ be a chordless 2-factor in $G$, and let $F_2$ be a 2-factor that is locally optimal with respect to $F_1$. 
Then for any cycle $C$ in $(V, F_1)$, there exists a cycle $D$ in $(V, F_2)$ of size at least 10 such that $|V(C)\cap V(D)|\ge 4$.
\end{lemma}

\begin{proof}
Suppose $F_2$ is locally optimal with respect to $F_1$, and assume by contradiction that there is some cycle $C$ in $(V, F_1)$ such that $(V, F_2)$ contains no cycle of size at least 10 that intersects $C$ in at least 4 nodes. 
Let $F_2'=F_2\triangle E(C)$; we will show that $(V, F_2')$ has fewer components than $(V, F_2)$, contradicting the fact that $F_2$ is locally optimal with respect to $F_1$.

Consider an arbitrary cycle $D$ in $(V, F_2)$ that intersects $C$. We will first show that any node $v$ in $D$ will be in a cycle in $(V, F_2')$ that is at least as large as $D$. This shows that the number of cycles of $(V, F_2')$ is at most the number of cycles of $(V, F_2)$. We then show that it is not possible that for every node, its cycle in $(V, F_2')$ is the same size as the cycle containing it in $(V, F_2)$.

If $D$ contains exactly one edge, say $e$, in $C$, then in $(V, F_2')$, the edge $e$ in $D$ is replaced by an odd-length path. Hence, in this case the nodes in $D$ will be contained in a cycle in $(V, F_2')$ that is strictly larger than $D$. 

If $D$ contains $k>1$ edges in $C$, then $D$ has size at most 8, since otherwise $D$ contradicts our assumption that no cycle exists in $(V, F_2)$ of size at least 10 that intersects $C$ in at least 4 nodes. We now show this implies that $D$ has size exactly 8 and $k=2$:
The size of $D$ is either 6 or 8 since $G$ has no potential 4-cycles. Note that $D$ alternates edges in $C$ and odd-length paths in $(V, F_2\setminus E(C))$. Since $C$ is chordless, the paths cannot have length 1 and must thus have length at least 3.
We thus have that $D$ must consist of exactly two edges from $C$, say $e_1, e_2$, separated by two paths of length 3, say $P_1, P_2$.

Since $P_1$ and $P_2$ do not contain edges in $C$, $(V, F_2')$ also contains the edges in $P_1$ and $P_2$. Hence, to show that all nodes in $D$ are in cycles of size at least 8 in $(V, F_2')$, it now suffices to show that the cycles containing $P_1$ and $P_2$ in $(V, F_2')$ have size at least 8.
Consider the cycle in $(V, F_2')$ containing $P_1$; besides $P_1$, the cycle contains another path, say $P_1'$, connecting the endpoints of $P_1$, and this path must have odd length $\ge 3$ since $G$ is bipartite and has no potential 4-cycles. Furthermore, $P_1'$ starts and ends with an edge in $C$, by definition of $F_2'=F_2\triangle E(C)$. Note that $P_1'$ thus cannot have length 3, as this would imply that the middle edge in $P_1'$ is a chord for $C$. So $P_1'$ has length at least 5, and the cycle in $(V, F_2')$ containing $P_1$ thus has size at least 8. Similarly, the cycle containing $P_2$ in $(V, F_2')$ has size at least 8.

We have thus shown that all nodes in $D$ are in cycles of size at least $|V(D)|$ in $(V, F_2')$, and hence, $(V, F_2')$ has at most as many cycles as $(V, F_2)$. 
Furthermore, it follows from the argument given above that the number of cycles in $(V, F_2)$ and $(V, F_2')$ is the same only if all nodes in $C$ are in cycles of size 8 in both $(V, F_2)$ and $(V, F_2')$ and each such cycle consists of two edges from $E(C)$ and two paths of length 3 in $(V, F_2\setminus E(C))$.

We now show by contradiction that the latter is impossible. Suppose $C$ is such that both in $(V, F_2)$ and in $(V, F_2')$ every node in $V(C)$ is contained in a cycle containing two edges from $C$. Then $|V(C)|$ must be a multiple of 4, say $|V(C)|=4k$.
Let the nodes of $C$ be labeled $1,2, \ldots, 4k\equiv 0$, such that $\{2i+1,2i+2\}\in F_2$ and $\{2i, 2i+1\}\in F_2'$ for $i=0, \ldots, 2k-1$.
We also define a mapping $p(i)$ for every $i=1,\ldots, 4k$, such that $(V, F_2\setminus E(C))$ contains a path (which, by our assumption has length 3) from $i$ to $p(i)$ for $i=1,\ldots, 4k$. 
Observe that by the definition of the mapping, $p(p(i))$ must be equal to $i$ (mod $4k$) for $i=1,\ldots, 4k$.

Let $p(1)=\ell$, then the fact that edge $\{1,2\}$ is in a cycle with one other edge from $C$ in $(V, F_2)$ implies that either $\{\ell, \ell+1\}\in F_2$ or $\{\ell, \ell-1\}\in F_2$, and that either $p(2)=\ell+1$ or $p(2)=\ell-1$.
In the first case, edge $\{2,3\}$ must be in a cycle with $\{\ell+1, \ell+2\}$ in the second 2-factor $(V, F_2')$, and thus $p(3)=\ell+2$. In the second case, $\{2,3\}$ must  be in a cycle with $\{\ell-1, \ell-2\}$ in $(V, F_2')$, and thus $p(3)=\ell-2$. 
Repeating the argument shows that either $p(i)\equiv\ell+(i-1)$ (mod $4k$) for $i=1, \ldots, 4k$, or $p(i)\equiv\ell-(i-1)$ (mod $4k$) for $i=1,\ldots, 4k$.

The first case gives a contradiction to the fact that $G$ is bipartite: note that $p(p(i))\equiv 2\ell+i-2$ (mod $4k$) and this must be equal to $i$. Hence, $\ell \equiv 1$ (mod $2k$); in other words, $\ell$ is odd, which cannot be the case since $(V, F_2\setminus E(C))$ has a path of length 3 from $1$ to $\ell$ (since $p(1)=\ell$) and if $\ell$ were odd, then $C$ would have a path of even length from $1$ to $\ell$, and thus $G$ would contain an odd cycle.

Now suppose that $p(i)\equiv \ell-(i-1)$ (mod $4k$) for $i=1,\ldots, 4k$. From the previous argument we know that $\ell$ must be even, since otherwise $G$ is not bipartite. But then $p(\ell/2)=\ell-(\ell/2-1)=\ell/2+1$. In other words, node $\ell/2$ is connected to node $\ell/2+1$ in $(V, F_2\setminus E(C))$. But since the edge $\{\ell/2,\ell/2+1\}$ is either in $F_2$ or in $F_2'$, $\{\ell/2,\ell/2+1\}$ is the only edge from $C$ in its cycle in either $(V, F_2)$ or $(V, F_2')$, contradicting the assumption on $C$.
\end{proof}

It may be the case that Lemma~\ref{lemma:chordless} also holds for cycles in $(V,F_1)$ that do have chords, but we have not been able to prove this. Instead, there is a simple alternative operation that ensures that a cycle in $(V, F_1)$ with a chord intersects at least one ``long'' cycle of size at least 10 in $(V, F_2)$ in 4 or more nodes. The algorithm described next will add this operation, and for technical reasons it will only modify $F_2$ with respect to a cycle $C$ in $(V, F_1)$ if the cycle $C$ does not yet intersect a long cycle in $(V, F_2)$ in 4 or more nodes.

\subsection{A 2-Factor with Average Cycle Size 8}\label{tech}
We give our algorithm in Algorithm~\ref{fig:cubicbip}. 
The algorithm fixes a 2-factor $F_1$ and initializes $F_2$ to be a 2-factor that contains all edges in $E\setminus F_1$. The algorithm then proceeds to modify $F_2$; note that $F_1$ is not changed.
Figure~\ref{fig:chord} illustrates the modification to $F_2$ in the case of a chorded cycle $C_i$.

\begin{algorithm}[!h]
\caption{Approximation Algorithm for Cubic Bipartite TSP}\label{fig:cubicbip}
{
Let $G=(V,E)$ be a bipartite cubic graph, with potential 4-cycles contracted using Lemma~\ref{lemma:four}.\;

Let $F_1$ be an arbitrary 2-factor in $G$, and let $C_1, \ldots, C_k$ be the cycles in $(V, F_1)$.\;

For each cycle $C_i$ in $(V,F_1)$, let $M(C_i)\subseteq E(C_i)$ be a perfect matching on $V(C_i)$.\;

Initialize $F_2 = (E\backslash F_1)\cup \bigcup_{i=1}^k M(C_i) $.\;

\While{there exists a cycle $C_i$ such that $|V(C_i)\cap V(D)|<4$ for all cycles $D$ in $(V, F_2)$ of size at least 10}
{
	\eIf {$C_i$ is a chordless cycle}
		{
		$F_2\leftarrow F_2\triangle E(C_i)$.\;
		}	
		{
		Let $\{x, y\}$ be a chord for $C_i$, let $P_1,P_2$ be the edge disjoint paths in $C_i$ from $x$ to $y$.\;
		Relabel $P_1$ and $P_2$ if necessary so that $P_1$ starts and ends with an edge in $F_1\setminus F_2$.\;
		$F_2\leftarrow \left(F_2\triangle E(P_1)\right) \setminus \{x,y\}$.\;\label{alg:pivot}
		}
}
Uncontract the 4-cycles in $(V,F_j)$ for $j=1,2$ using Lemma~\ref{lemma:four}.\;
Return the 2-factor among $F_1, F_2$ with the smaller number of components.\;
}		
\end{algorithm}
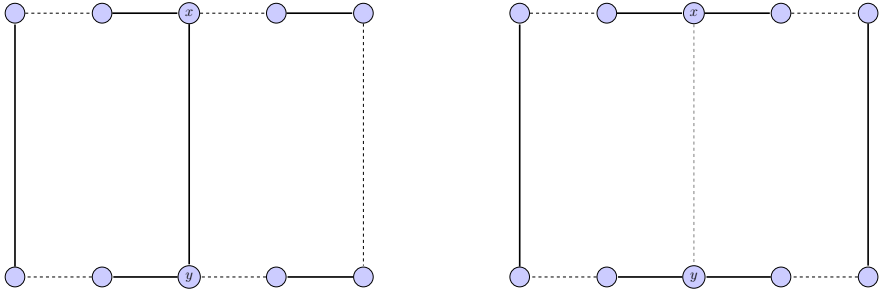
\begin{figure}[h!]\begin{center}
\resizebox{1.95in}{1.5in}{
\begin{tikzpicture}[->,>=stealth',shorten >=1pt,auto,node distance=3cm,
  thick,main node/.style={circle,fill=blue!20,draw,font=\sffamily\Large\bfseries},aux node/.style={}]

  \node[main node] (1) {$\ \ $};
   \foreach \name/\pos/\label in {{2/1/\ \ }, {3/2/x}, {4/3/\ \ }, {5/4/\ \ }}
       \node[main node] (\name) [right of = \pos] {$\label$};

  \foreach \name/\pos/\label in {{2a/2/}, {2b/2a/}, {4a/4/}, {4b/4a/}}
       \node[aux node] (\name) [below of = \pos] {$\label$};

  \foreach \name/\pos/\loc/\label in {{9/2b/below of/\ \ }, {7/4b/below of/\ \ }, {10/9/left of/\ \ }, {8/7/left of/y}, {6/7/right of/\ \ }}
       \node[main node] (\name) [\loc = \pos] {$\label$};
       
\path[-,every node/.style={font=\sffamily\Large}]
						
	\foreach \source/ \dest in {1/2, 3/4,5/6,7/8,9/10} {
		(\source) edge [dashed] node  {} (\dest)
	}
						
	\foreach \source/ \dest in {2/3,4/5,6/7,8/9,10/1} {
		(\source) edge [ultra thick] node  {} (\dest)
	}
	
	(3) edge [ultra thick] node {$$} (8);
;
\end{tikzpicture}}
\hspace*{1.5cm}
\resizebox{1.95in}{1.5in}{
\begin{tikzpicture}[->,>=stealth',shorten >=1pt,auto,node distance=3cm,
  thick,main node/.style={circle,fill=blue!20,draw,font=\sffamily\Large\bfseries},aux node/.style={}]

  \node[main node] (1) {$\ \ $};
   \foreach \name/\pos/\label in {{2/1/\ \ }, {3/2/x}, {4/3/\ \ }, {5/4/\ \ }}
       \node[main node] (\name) [right of = \pos] {$\label$};

  \foreach \name/\pos/\label in {{2a/2/}, {2b/2a/}, {4a/4/}, {4b/4a/}}
       \node[aux node] (\name) [below of = \pos] {$\label$};

  \foreach \name/\pos/\loc/\label in {{9/2b/below of/\ \ }, {7/4b/below of/\ \ }, {10/9/left of/\ \ }, {8/7/left of/y}, {6/7/right of/\ \ }}
       \node[main node] (\name) [\loc = \pos] {$\label$};
       
\path[-,every node/.style={font=\sffamily\Large}]

	\foreach \source/ \dest in {2/3,3/4,5/6,7/8,8/9,10/1} {
		(\source) edge [ultra thick] node  {} (\dest)
	}

	\foreach \source/ \dest in {1/2,4/5,6/7,9/10} {
		(\source) edge [dashed]  node  {} (\dest)
	}
	(3) edge [draw=gray, very thin, dashed] node {$$} (8);
;
\end{tikzpicture}}
\end{center}
\caption{The figure on the left shows a chorded cycle $C$ in $F_1$ of size 10 and all edges in $G$ that have both endpoints in $C$. The dashed edges are in $F_1\setminus F_2$, and non-dashed edges are in $F_2$ (where not all edges in $F_2$ that are incident on the nodes are shown). The figure on the right shows how Algorithm~\ref{fig:cubicbip} would update $F_2$. 
}\label{fig:chord}
\end{figure}

We need to prove that the set of edges $F_2$ remains a 2-factor throughout the course of the algorithm, that the algorithm terminates, and that upon termination, either $(V, F_1)$ or $(V, F_2)$ has at most $|V|/8$ components. 
The latter is clear: {\it if} the algorithm terminates, then the condition of Lemma~\ref{lemma:avg} is satisfied, and therefore one of the two 2-factors has at most $|V|/8$ components.

To show that $F_2$ is a 2-factor and that the algorithm terminates is a little more subtle. In order to show this, it will be helpful to know that each cycle in $(V, F_i)$ alternates edges in $F_i\cap F_{i+1}$ and edges in $F_i\setminus F_{i+1}$ for $i=1,2$ (where subscripts are modulo 2, so $F_3\equiv F_1$). This is true initially, however, it is not the case that this property continues to hold for all cycles.
We will show that it does hold in certain cases, which turn out to be exactly the cases ``when we need it''. In the following, we will say a cycle or path in $(V, F_i)$ is {\it alternating} (for $F_1$ and $F_2)$ if it alternates edges in $F_i\cap F_{i+1}$ and $F_i\setminus F_{i+1}$.
We will say that a cycle $C$ in $(V, F_1)$ is {\it violated} if there exists no $D$ of size at least 10 in $(V, F_2)$ such that $|V(C)\cap V(D)|\ge 4$.

\begin{lemma}\label{lemma:invariant}
Algorithm~\ref{fig:cubicbip} maintains that $F_2$ is a 2-factor that satisfies the following properties:
\begin{enumerate}
\item[(1)]
if $C$ in $(V, F_1)$ is violated, then $C$ is alternating for $F_1$ and $F_2$;
\item[(2)]
if $D$ in $(V, F_2)$ is not alternating for $F_1$ and $F_2$, then $D$ has size at least 10.
\end{enumerate}
\end{lemma}
\begin{proof}
We prove the lemma by induction on the algorithm. Initially, $F_2$ consists of $E\setminus F_1$ and 
$\bigcup_{i=1}^k M(C_i)$, which are two edge-disjoint perfect matchings on $V$. Hence, $F_2$ is a 2-factor, and the two properties hold for all cycles in $(V, F_1)$ and $(V, F_2)$.
\newtheorem{myclaim}{Claim}

Suppose the lemma holds and we modify $F_2$ by considering some violated cycle $C$. The two properties of the lemma imply the following:
\begin{myclaim}\label{claim}
If $C$ is violated, then $(V, F_2\setminus E(C))$ consists of even cycles and odd-length paths, where paths that are not alternating for $F_1$ and $F_2$ have length at least 9.
\end{myclaim}
\begin{proof_of_claim}
For each path in the graph $(V, F_2\setminus E(C))$ there exists some cycle $D$ in $(V, F_2)$ such that the path results when removing $E(C)\cap E(D)$ from $D$. If $D$ is alternating for $F_1$ and $F_2$, then the path must have the same property, and it must start and end with an edge in $F_2\setminus F_1$. Hence, the path must have odd length if $D$ is alternating.
If $D$ is not alternating then $D$ has size at least 10 by Property (2), so $C$ can have at most one edge in common with $D$, since otherwise $C$ is not violated. Hence, the path obtained by removing the unique edge in $E(C)\cap E(D)$ has length at least 9, and its length must be odd, since $D$ is an even cycle.
\end{proof_of_claim}\\

If $C$ is chordless, then we modify $F_2$ to $F_2'=F_2\triangle E(C)$. Clearly, $F_2'$ is again a 2-factor, and Property (1) remains satisfied.
Furthermore, any cycle in $(V, F_2')$ that is not alternating for $F_1$ and $F_2'$ either also existed in $(V, F_2)$ and hence it has size at least 10, since Property (2) holds for $F_2$, or the cycle contains a path in $(V, F_2\setminus E(C))$ that is not alternating for $F_1$ and $F_2$, and this path has length at least 9 by the claim. So Property (2) holds for $F_2'$.

Now consider the modification of $F_2$ when considering a chorded cycle $C$ in $(V, F_1)$. Let $P_1$ be as defined in the algorithm. First consider $F_2'=F_2\triangle E(P_1)$; every node is incident to two edges in $F_2'$, except for $x$ and $y$, which are incident to three edges in $F_2'$, namely two edges in $E(C)$ plus the edge $\{x,y\}$. Hence, removing $\{x,y\}$ will give a new 2-factor, say $F_2''$. The modification from $F_2$ to $F_2''$ is exactly the modification made to $F_2$ by the algorithm.

We now show that the two properties are satisfied.
Clearly, $C$ is not alternating for $F_1$ and $F_2''$, so in order to maintain Property (1), we need to show that $C$ is no longer violated. To do this, we show that $(V, F_2'')$ contains a cycle of size at least 10 that contains $x$, $y$ and their 4 neighbors in $C$. 
First, suppose by contradiction that after removing $\{x,y\}$, $x$ and $y$ are not in the same cycle. Consider the component of $(V, F_2'')$ containing $x$: starting from $x$, it alternates edges in $E(C)$ and paths in $F_2\setminus E(C)$, starting and ending with an edge in $E(C)$. By Claim~\ref{claim} the paths in $(V, F_2\setminus E(C))$ have odd length, and hence the component containing $x$ must be an odd cycle, contradicting the fact that $G$ is bipartite. So, $x$ and $y$ must be in the same cycle in $(V, F_2'')$.
This cycle must thus consist of two odd-length paths from $x$ to $y$, each starting and ending with an edge in $F_2''\cap F_1$. These paths cannot have length $3$, because this would imply that the path plus the edge $\{x,y\}$ would form a potential 4-cycle. Hence, the cycle in $(V, F_2'')$ containing $x$ and $y$ has size at least 10. 

For Property (2), note that any cycle $D$ in $(V, F_2'')$ that is not alternating for $F_1$ and $F_2''$ either (i) existed in $(V, F_2)$ and therefore has size at least 10, or (ii) contains $x$ and $y$ and we showed above that this cycle has size at least 10, or (iii) it contains a path in $(V, F_2\setminus E(C))$ that is not alternating for $F_1$ and $F_2$, and by the claim this path has length at least 9. Hence, Property (2) is satisfied by $F_2''$.
\end{proof}

\begin{lemma}\label{lemma:2factor}
Given a cubic bipartite graph $G=(V,E)$, Algorithm~\ref{fig:cubicbip} returns a 2-factor in $G$ with at most $|V|/8$ components.
\end{lemma}
\begin{proof}
By Lemma~\ref{lemma:four} it suffices to show that the current lemma holds if $G$ has no potential 4-cycles. By the termination condition of Algorithm~\ref{fig:cubicbip} and Lemma~\ref{lemma:avg}, the 2-factor returned by the algorithm does indeed have at most $|V|/8$ components, so it remains to show that the algorithm always returns a 2-factor.

By Lemma~\ref{lemma:invariant}, the algorithm maintains two 2-factors $F_1$ and $F_2$.
Observe that if a cycle $C$ is not violated, then this continues to hold throughout the remainder of the algorithm: 
Let $D$ be a cycle of size at least 10 in $(V, F_2)$ such that $|V(C)\cap V(D)|\ge 4$. The only possible changes to $D$ will be caused by a violated cycle $C'$, which necessarily contains at most one edge in $D$: by Lemma~\ref{lemma:invariant} $C'$ is alternating for $F_1$ and $F_2$, so if $C'$ contains more than one edge in $D$, $C'$ cannot be violated. The modification of $F_2$ with respect to $E(C')$ can therefore only cause the cycle $D$ to become a larger cycle $D'$ where $V(D')\supseteq V(D)$. So $D'$ will have size at least 10, and $|V(C)\cap V(D')|\ge 4$.

It remains to show that if we modify $F_2$ with respect to some violated cycle $C$, then $C$ is not violated for the new 2-factor $F_2'$. If $C$ is not chordless, then this holds because $C$ is not alternating for $F_1$ and the new 2-factor $F_2'$, so by Lemma~\ref{lemma:invariant}, $C$ is not violated.
If $C$ is chordless and violated, then by Claim~\ref{claim}, $(V, F_2\setminus E(C))$ consists of even cycles and odd-length paths. The proof of Lemma~\ref{lemma:chordless} then shows that taking the symmetric difference of $F_2$ with $E(C)$ (strictly) reduces the number of components. This implies that for $F_2'=F_2\triangle E(C)$, cycle $C$ is not violated: otherwise, we could apply the same arguments to show that $(V, F_2\triangle E(C)\triangle E(C))$ has strictly fewer cycles than $(V, F_2)$, but this is a contradiction since $F_2\triangle E(C)\triangle E(C)=F_2$.
\end{proof}

In the appendix, we give an example on 48 nodes that shows our analysis of Algorithm~\ref{fig:cubicbip} is tight.
In fact, the example is also tight for the local improvement heuristic from Section~\ref{local} and for the local improvement heuristic we obtain if we allow Algorithm~\ref{fig:cubicbip} to modify $F_2$ for cycles $C$ that are chorded and/or do have at least two edges in a cycle of size 10 or more in $(V, F_2)$.

\section{Cubic graphs}\label{sec:cub}
We now consider cubic graphs, in other words, we drop the requirement that the graph is bipartite. We assume the graph is 2-connected.
The best known approximation result for graph-TSP on a 2-connected cubic graphs $G=(V,E)$ is due to Correa, Larr\'e and Soto~\cite{CorreaLS15} who show how to find in polynomial time a tour of length at most $\left(\frac43-\frac1 {61236}\right)|V|$.

One obstacle for their techniques are chorded 4-cycles, i.e., a set of 4 nodes $(v_1,v_2,v_3,v_4)$ such that the subgraph of $G$ induced by $\{v_1,v_2,v_3,v_4\}$ contains edges $\{v_1,v_2\}, \{v_2,v_3\}$,
$\{v_3,v_4\}, \{v_4,v_1\}$ and the ``chord'' $\{v_2,v_4\}$.
In fact, Correa et al. prove the following.

\begin{lemma}[Correa, Larr\'e and Soto~\cite{CorreaLS15}]\label{correals}
Consider a graph-TSP instance on a 2-connected cubic graph $G=(V,E)$, and let $B$ be the set of nodes in $G$ contained in a chorded 4-cycle. Then, we can find in polynomial time a tour of length at most $\frac43|B| + (\frac43-\frac1{8748})(|V\backslash B|) +2$.
\end{lemma}
The proof of this lemma is contained in the proof of Theorem 2 in~\cite{CorreaLS15}. More precisely, it is shown that there exists a distribution over tours of expected size $\sum_{v\in V} z(v)+2$, where $z(v)$ is the ``average contribution'' of $v$. In the proof of the Theorem 2 in \cite{CorreaLS15} it is shown that $\sum_{v\in V} z(v)\le \frac43|B| + (\frac43-\frac1{8748})(|V\backslash B|)$.

On the other hand, chorded 4-cycles are ``beneficial'' for the analysis of the M\"omke-Svensson~\cite{MomkeS16} algorithm for the graph-TSP on 2-connected subcubic graphs, as we will show next. 

\begin{lemma}\label{highb}
Consider a graph-TSP instance on a 2-connected subcubic graph $G=(V,E)$, and let $B$ be the set of nodes in $G$ contained in a chorded 4-cycle. Then, we can find in polynomial time a tour of length at most $\frac43|V| - \frac16|B| -\frac23$.
\end{lemma}
\begin{proof}
For convenience, we will use the word ``tour'' to refer to a connected spanning Eulerian multigraph obtained by doubling and deleting some edges of $G$. Note that the graph-TSP instance on $G$ indeed has a tour (obtained by shortcutting a Eulerian walk on this multigraph) of length at most the number of edges in the multigraph. M\"omke and Svensson~\cite{MomkeS16} show that there exists a probability distribution over tours, such that
each edge $e\in E$ appears an even number of times (zero or twice) with probability $\frac13$, and the expected number of edges in the tour is $\frac43|V|-\frac23$.

Now, we can simply contract chorded 4-cycles in a cubic graph to obtain a subcubic graph $\tilde G$, on which we can apply the M\"omke-Svensson algorithm to find a distribution over tours with $\frac43(|V|-\frac34|B|)-\frac23=\frac43|V|-|B|-\frac23$ edges in expectation. 
Next, we uncontract the chorded 4-cycles in each of the tours $T$:
For a chorded 4-cycle $(v_1,v_2,v_3,v_4)$ that is contracted to a node $v$, note that $T$ contains an even number of edges incident to $v$. Hence, either both edges incident to $v$ in $\tilde G$ appear an odd number of times in $T$, or both appear an even number of times. 
In the first case, we uncontract the chorded 4-cycle by adding the edges $\{v_1,v_2\}, \{v_2,v_4\}$ and $\{v_3,v_4\}$.
In the second case, we uncontract the chorded 4-cycle by adding the edges $\{v_1,v_2\}, \{v_2,v_3\}, \{v_3,v_4\}, \{v_4,v_1\}$.

Since each edge appears an even number of times with probability $\frac13$, the second case happens with probability $\frac 13$, and the first case thus happens with probability $\frac 23$. Hence, the expected number of edges in $T$ increases by $\frac {10}3 \frac{|B|}4$, giving a total number of edges of $\frac43|V|-\frac16|B|-\frac23$ in expectation.
\end{proof}

Note that the bound in Lemma~\ref{correals} is increasing in $|B|$ and the bound in Lemma~\ref{highb} is decreasing in $|B|$.
Setting the bound in Lemma~\ref{correals} equal to the bound in Lemma~\ref{highb} gives $|B| = \frac1{1459}|V|$
and thus shows that there exists a polynomial time algorithm for finding a tour of length at most $\left(\frac43- \frac1{8754}\right)|V|$ for a graph-TSP instance on a 2-connected cubic graph $G=(V, E)$. By an observation of M\"omke and Svensson~\cite{MomkeS16}, this also implies a $\left(\frac43- \frac1{8754}\right)$-approximation algorithm for cubic graph-TSP, i.e., the graph $G$ does not have to be 2-connected.
We have thus shown the following result.
\begin{theorem}
There exists a $\left(\frac43- \frac1{8754}\right)$-approximation algorithm for Graph-TSP on cubic graphs.
\end{theorem}


\subsection*{Acknowledgements}
The author would like to thank Marcin Mucha for careful reading and pointing out an omission in a previous version, Frans Schalekamp for helpful discussions, and an anonymous reviewer for suggesting the simplified proof for the result in Section~\ref{sec:cub} for cubic non-bipartite graphs. Other anonymous reviewers are acknowledged for helpful feedback on the presentation of the algorithm for bipartite cubic graphs.

\section{Tightness of the Analysis of Algorithm~\ref{fig:cubicbip}}\label{sec:tight}
We give an example of a cubic bipartite graph $G=(V, E)$ for which both the 2-factors $F_1$ and $F_2$ that result from Algorithm~\ref{fig:cubicbip} have $|V|/8$ components.

The instance has $48$ nodes, numbered $1$ through $48$, and $(V, F_1)$ contains six cycles, four cycles of size 6, one cycle of size 10 and one cycle of size 14. For brevity, we denote the cycles by only giving an ordered listing of their nodes; an edge between consecutive nodes and between the last and first node is implicit. The cycles in $(V, F_1)$ are:
\begin{align*}
&C_1=(1,2,3,4,5,6), \\
&C_2=(7,8,9,10,11,12), \\
&C_3=(13,14,15,16,17,18), \\
&C_4=(19,20,21,22,23,24), \\
&C_5=(25,26,27,28,29,30,31,32,33,34), \\
&C_6=(35,36,37,38,39,40,41,42,43,44,45,46,47,48).
\end{align*}

The second 2-factor $(V,F_2)$ has six cycles as well, namely, three cycles of size 6 and three cycles of size 10. 
We again denote the cycles by giving an ordered listing of the nodes, but now a semicolon between subsequent nodes indicates that the nodes are connected by an edge in $F_2\setminus F_1$, and a comma denotes that they are connected by an edge in $F_2\cap F_1$.
The cycles in $(V, F_2)$ are
\begin{align*}
&D_1 =(6,1;30,31;42,43;28,29;40,41;), &&D_4=(18,13;38,39;44,45;), \\
&D_2=(48,35;10,11;4,5;8,9;2,3;), &&D_5=(24,19;26,27;32,33;), \\
&D_3=(34,25;16,17;22,23;14,15;20,21;), &&D_6=(12,7;36,37;46,47;).
\end{align*}
It is straightforward to verify that every node occurs in exactly one cycle in $(V, F_1)$ and exactly one cycle in $(V, F_2)$, and that each cycle in $(V, F_1)$ alternates edges in $F_1\setminus F_2$ and edges in $F_1\cap F_2$, and that each cycle in $(V, F_2)$ alternates edges in $F_2\cap F_1$ and edges in $F_2\setminus F_1$. Furthermore, each cycle $C_i$ in $(V, F_1)$ has exactly two edges in a cycle $D$ in $(V, F_2)$ of size exactly 10.

\subsubsection*{Local Optimum}
Figure~\ref{fig:example} depicts each of the six cycles $C_i$ in $(V, F_1)$, together with the cycles $D_j$ in $(V,F_2)$ that intersect the given cycle $C_i$. For any of the cycles $C_i$, replacing $F_2$ by $F_2\triangle E(C_i)$ does not decrease the number of components of $(V, F_2)$ for any cycle $C_i$, nor does the modification of $F_2$ for a chorded cycle described in Algorithm~\ref{fig:cubicbip}.

Hence, if, rather than following Algorithm~\ref{fig:cubicbip}, we would execute one of the two possible modifications of $F_2$ suggested by the algorithm, as long as this reduced the number of components of $(V, F_2)$, then the 2-factor $F_2$ is in fact a local optimum with respect to this process since none of the possible modifications reduces the number of components.

We note that the instance does have a Hamilton cycle. This cycle contains subpaths of more than two adjacent edges in $E\setminus F_1$, so it can never be found from $F_1$ and $F_2$ using the ``moves'' we defined, even if we allow ``moves'' that do not reduce the number of components of $(V, F_2)$.
We give the Hamilton cycle by again giving an ordered listing of the nodes, where a semicolon between subsequent nodes indicates that the nodes are connected by an edge in $E\setminus F_1$, and a comma denotes that they are connected by an edge in $F_1$.
\begin{align*}
&(1;30, 29 ; 40,39,38 ; 13,14,15,16;25,26;19,20,21;34,33;24,23,22;17,18;45,\\
&44,43;28,27;32,31;42,41;6,5;8,7,12,	11;	4,3;48,47,46;37,36,35;10,9;2,1).
\end{align*}

\newlength{\x}
\setlength{\x}{0.000566\textwidth}
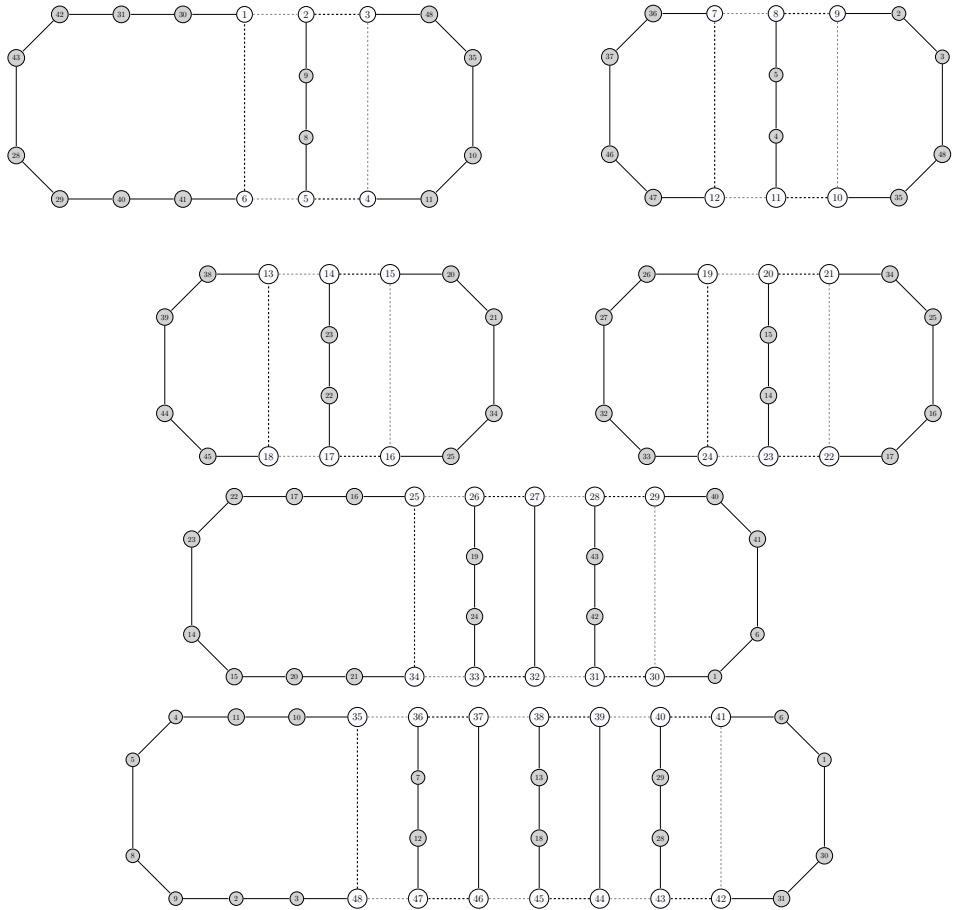
\begin{figure}[h!]\begin{center}
 \resizebox{886.82\x}{!}{
 \begin{tikzpicture}[->,>=stealth',shorten >=1pt,auto,node distance=3cm,
  thick,main node/.style={circle,fill=blue!2,draw,font=\sffamily\Large\bfseries},aux node/.style={circle,fill=gray!35,draw,font=\sffamily\bfseries}]

  \node[main node] (1) {$1$};
  
   \foreach \name/\pos in {{2/1}, {3/2}}
       \node[main node] (\name) [right of = \pos] {$\name$};

  \foreach \name/\pos in {{9/2}, {8/9}}
       \node[aux node] (\name) [below of = \pos] {$\name$};
       
  \foreach \name/\pos/\loc in {{5/8/below of}, {6/5/left of}, {4/5/right of}}
       \node[main node] (\name) [\loc = \pos] {$\name$};
       
  \foreach \name/\pos/\loc in {{48/3/right of},{35/48/below right of},{11/4/right of},{10/11/above right of}}
       \node[aux node] (\name) [\loc = \pos] {$\name$};

   \foreach \name/\pos/\loc in {{30/1/left of},{31/30/left of},{42/31/left of},{43/42/below left of},
   {41/6/left of},{40/41/left of},{29/40/left of},{28/29/above left of}}
          \node[aux node] (\name) [\loc = \pos] {$\name$};

\path[-,every node/.style={font=\sffamily\Large}]
	\foreach \source/\dest in {1/2,3/4,5/6}{
		(\source) edge [draw=gray, dashed] node {} (\dest)	}
	\foreach \source/\dest in {2/3,4/5,6/1}{
		(\source) edge [ultra thick, dashed] node {} (\dest)	}

	\foreach \source/ \dest in {2/9,9/8,8/5,3/48,48/35,35/10,10/11,11/4,1/30,30/31,31/42,42/43,43/28,28/29,29/40,40/41,41/6} {
		(\source) edge [very thick] node  {} (\dest)
	}
;
\end{tikzpicture}
}
 \hspace{1.3cm}  \resizebox{655.81\x}{!}{
 \begin{tikzpicture}[->,>=stealth',shorten >=1pt,auto,node distance=3cm,
  thick,main node/.style={circle,fill=blue!2,draw,font=\sffamily\Large\bfseries},aux node/.style={circle,fill=gray!35,draw,font=\sffamily\bfseries}]

  \node[main node] (7) {$7$};
  
   \foreach \name/\pos in {{8/7}, {9/8}}
       \node[main node] (\name) [right of = \pos] {$\name$};

  \foreach \name/\pos in {{5/8}, {4/5}}
       \node[aux node] (\name) [below of = \pos] {$\name$};
       
  \foreach \name/\pos/\loc in {{11/4/below of}, {12/11/left of}, {10/11/right of}}
       \node[main node] (\name) [\loc = \pos] {$\name$};
       
  \foreach \name/\pos/\loc in {{2/9/right of},{3/2/below right of},{35/10/right of},{48/35/above right of}}
       \node[aux node] (\name) [\loc = \pos] {$\name$};

   \foreach \name/\pos/\loc in {{36/7/left of},{37/36/below left of},{47/12/left of},{46/47/above left of}}
          \node[aux node] (\name) [\loc = \pos] {$\name$};

\path[-,every node/.style={font=\sffamily\Large}]
	\foreach \source/\dest in {7/8,9/10,11/12}{
		(\source) edge [draw=gray, dashed] node {} (\dest)	}
	\foreach \source/\dest in {8/9,10/11,12/7}{
		(\source) edge [ultra thick, dashed] node {} (\dest)	}

	\foreach \source/ \dest in {9/2,2/3,3/48,48/35,35/10,8/5,5/4,4/11,7/36,36/37,37/46,46/47,47/12} {
		(\source) edge [very thick] node  {} (\dest)
	}
;
\end{tikzpicture}
}
  \vspace{.5\baselineskip}
  
\hspace{1.85cm}  \resizebox{649\x}{!}{
\begin{tikzpicture}[->,>=stealth',shorten >=1pt,auto,node distance=3cm,
  thick,main node/.style={circle,fill=blue!2,draw,font=\sffamily\Large\bfseries},aux node/.style={circle,fill=gray!35,draw,font=\sffamily\bfseries}]

  \node[main node] (13) {$13$};
  
   \foreach \name/\pos in {{14/13}, {15/14}}
       \node[main node] (\name) [right of = \pos] {$\name$};

  \foreach \name/\pos in {{23/14}, {22/23}}
       \node[aux node] (\name) [below of = \pos] {$\name$};
       
  \foreach \name/\pos/\loc in {{17/22/below of}, {18/17/left of}, {16/17/right of}}
       \node[main node] (\name) [\loc = \pos] {$\name$};
       
  \foreach \name/\pos/\loc in {{20/15/right of},{21/20/below right of},{25/16/right of},{34/25/above right of}}
       \node[aux node] (\name) [\loc = \pos] {$\name$};

   \foreach \name/\pos/\loc in {{38/13/left of},{39/38/below left of},{45/18/left of},{44/45/above left of}}
          \node[aux node] (\name) [\loc = \pos] {$\name$};

\path[-,every node/.style={font=\sffamily\Large}]
	\foreach \source/\dest in {13/14,15/16,17/18}{
		(\source) edge [draw=gray, dashed] node {} (\dest)	}
	\foreach \source/\dest in {14/15,16/17,18/13}{
		(\source) edge [ultra thick, dashed] node {} (\dest)	}

	\foreach \source/ \dest in {15/20,20/21,21/34,34/25,25/16,14/23,23/22,22/17,13/38,38/39,39/44,44/45,45/18} {
		(\source) edge [very thick] node  {} (\dest)
	}
;
\end{tikzpicture} }\hspace{1.2cm}\resizebox{649\x}{!}{
\begin{tikzpicture}[->,>=stealth',shorten >=1pt,auto,node distance=3cm,
  thick,main node/.style={circle,fill=blue!2,draw,font=\sffamily\Large\bfseries},aux node/.style={circle,fill=gray!35,draw,font=\sffamily\bfseries}]

  \node[main node] (19) {$19$};
  
   \foreach \name/\pos in {{20/19}, {21/20}}
       \node[main node] (\name) [right of = \pos] {$\name$};

  \foreach \name/\pos in {{15/20}, {14/15}}
       \node[aux node] (\name) [below of = \pos] {$\name$};
       
  \foreach \name/\pos/\loc in {{23/14/below of}, {24/23/left of}, {22/23/right of}}
       \node[main node] (\name) [\loc = \pos] {$\name$};
       
  \foreach \name/\pos/\loc in {{34/21/right of},{25/34/below right of},{17/22/right of},{16/17/above right of}}
       \node[aux node] (\name) [\loc = \pos] {$\name$};

   \foreach \name/\pos/\loc in {{26/19/left of},{27/26/below left of},{33/24/left of},{32/33/above left of}}
          \node[aux node] (\name) [\loc = \pos] {$\name$};

\path[-,every node/.style={font=\sffamily\Large}]
	\foreach \source/\dest in {20/21,22/23,24/19}{
		(\source) edge [ultra thick, dashed] node {} (\dest)	}
	\foreach \source/\dest in {19/20,21/22,23/24}{
		(\source) edge [draw=gray, dashed] node {} (\dest)	}

	\foreach \source/ \dest in {21/34,34/25,25/16,16/17,17/22,20/15,15/14,14/23,19/26,26/27,27/32,32/33,33/24} {
		(\source) edge [very thick] node  {} (\dest)
	}
;
\end{tikzpicture}
}
 \vspace{.5\baselineskip}
 
  \resizebox{1089\x}{!}{
  \begin{tikzpicture}[->,>=stealth',shorten >=1pt,auto,node distance=3cm,
  thick,main node/.style={circle,fill=blue!2,draw,font=\sffamily\Large\bfseries},aux node/.style={circle,fill=gray!35,draw,font=\sffamily\bfseries}]

  \node[main node] (25) {$25$};
  
   \foreach \name/\pos in {{26/25}, {27/26},{28/27}, {29/28}}
       \node[main node] (\name) [right of = \pos] {$\name$};

  \foreach \name/\pos in {{19/26}, {24/19},{43/28},{42/43}}
       \node[aux node] (\name) [below of = \pos] {$\name$};
       
  \foreach \name/\pos/\loc in {{33/24/below of}, {34/33/left of}, {32/33/right of},{31/42/below of},{30/31/right of}}
       \node[main node] (\name) [\loc = \pos] {$\name$};
       
  \foreach \name/\pos/\loc in {{40/29/right of},{41/40/below right of},{1/30/right of},{6/1/above right of}}
       \node[aux node] (\name) [\loc = \pos] {$\name$};

   \foreach \name/\pos/\loc in {{16/25/left of},{17/16/left of},{22/17/left of},{23/22/below left of},
   {21/34/left of},{20/21/left of},{15/20/left of},{14/15/above left of}}
          \node[aux node] (\name) [\loc = \pos] {$\name$};

\path[-,every node/.style={font=\sffamily\Large}]
	\foreach \source/\dest in {25/26,27/28,29/30,31/32,33/34}{
		(\source) edge [draw=gray, dashed] node {} (\dest)	}
	\foreach \source/\dest in {26/27,28/29,30/31,32/33,34/25}{
		(\source) edge [ultra thick, dashed] node {} (\dest)	}

	\foreach \source/ \dest in {27/32,28/43,43/42,42/31,26/19,19/24,24/33,29/40,40/41,41/6,6/1,1/30,25/16,16/17,17/22,22/23,23/14,14/15,15/20,20/21,21/34} {
		(\source) edge [very thick] node  {} (\dest)
	}
;
\end{tikzpicture}
}
 \vspace{.5\baselineskip}
 
  \resizebox{1321\x}{!}{
  \begin{tikzpicture}[->,>=stealth',shorten >=1pt,auto,node distance=3cm,
  thick,main node/.style={circle,fill=blue!2,draw,font=\sffamily\Large\bfseries},aux node/.style={circle,fill=gray!35,draw,font=\sffamily\bfseries}]

  \node[main node] (35) {$35$};
  
   \foreach \name/\pos in {{36/35}, {37/36},{38/37}, {39/38},{40/39},{41/40}}
       \node[main node] (\name) [right of = \pos] {$\name$};

  \foreach \name/\pos in {{7/36}, {12/7},{13/38},{18/13},{29/40},{28/29}}
       \node[aux node] (\name) [below of = \pos] {$\name$};
       
  \foreach \name/\pos/\loc in {{47/12/below of}, {48/47/left of}, {46/47/right of},{45/18/below of},{44/45/right of},{43/28/below of},{42/43/right of}}
       \node[main node] (\name) [\loc = \pos] {$\name$};
       
  \foreach \name/\pos/\loc in {{6/41/right of},{1/6/below right of},{31/42/right of},{30/31/above right of}}
       \node[aux node] (\name) [\loc = \pos] {$\name$};

   \foreach \name/\pos/\loc in {{10/35/left of},{11/10/left of},{4/11/left of},{5/4/below left of},
   {3/48/left of},{2/3/left of},{9/2/left of},{8/9/above left of}}
          \node[aux node] (\name) [\loc = \pos] {$\name$};

\path[-,every node/.style={font=\sffamily\Large}]
	\foreach \source/\dest in {35/36,37/38,39/40,41/42,43/44,45/46,47/48}{
		(\source) edge [draw=gray, dashed] node {} (\dest)	}
	\foreach \source/\dest in {36/37,38/39,40/41,42/43,44/45,46/47,48/35}{
		(\source) edge [ultra thick, dashed] node {} (\dest)	}

	\foreach \source/ \dest in {39/44,40/29,29/28,28/43,37/46,38/13,13/18,18/45,36/7,7/12,12/47,41/6,6/1,1/30,30/31,31/42,35/10,10/11,11/4,4/5,5/8,8/9,9/2,2/3,3/48} {
		(\source) edge [very thick] node  {} (\dest)
	}
;
\end{tikzpicture}
}
 \end{center}
\caption{Depicted are six figures, one for each cycle $C_i$ in $(V, F_1)$ for $i=1,\ldots, 6$. The white nodes are the nodes in the cycle $C_i$, and the dashed edges are the edges in $F_1\setminus E(C_i)$. The non-dashed and thick dashed edges are edges in $F_2$. For each $C_i$, the number of components of $(V, F_2)$ does not decrease by replacing $F_2$ by $F_2\triangle E(C_i)$.
}\label{fig:example}
\end{figure}

\end{document}